\documentclass[aps,rmp,reprint,longbibliography,twocolumn]{revtex4-1}
\usepackage{enumerate,appendix}
\usepackage{amsmath, amsthm, amssymb,commath,braket}
\usepackage{color,calc,graphicx}
\usepackage[usenames,dvipsnames,svgnames,table,cmyk,hyperref]{xcolor}

\usepackage[charter,cal=cmcal,sfscaled=false]{mathdesign}

\definecolor{oeawblue}{cmyk}{0.9,0.68,0,0}
\definecolor{iqoqiblue}{cmyk}{0.76,0.11,0,0}
\usepackage{hyperref}

\hypersetup{
  pdftitle={Entanglement and Nonlocality of 1D Macroscopic Systems},
  pdfstartview=Fit,
  pdfpagelayout=SinglePage,
  colorlinks,
  citecolor=oeawblue,
  linkcolor=oeawblue,
  urlcolor=oeawblue}
\usepackage[figure]{hypcap}

\def\be{\begin{equation}}
\def\ee{\end{equation}}
\def\bea{\begin{eqnarray}}
\def\eea{\end{eqnarray}}
\def\bma{\begin{mathletters}}
\def\ema{\end{mathletters}}
\def\eref#1{Eq.~\ref{#1}}

\def\0{\overline{0}}

\def\q0{\underline{0}}

\def\H{{\cal H}}

\def\T{{\cal T}}

\def\C{{\mathbb C}}
\def\id{{\mathbb I}}

\def\H{{\cal H}}

\def\tr{\mbox{tr}}
\def\one{\leavevmode\hbox{\small1\normalsize\kern-.33em1}}

\def\braket#1#2{\langle#1|#2\rangle}

\def\proj#1{\ket{#1}\!\bra{#1}}

\newtheorem{theo}{Theorem}

\def\id{{\mathbb I}}

\allowdisplaybreaks

\begin{document}

\title{Entanglement and Nonlocality in Infinite 1D Systems}
\author{Zizhu Wang, Sukhwinder Singh and Miguel Navascu\'es}
\affiliation{Institute for Quantum Optics and Quantum Information (IQOQI) Vienna, Austrian Academy of Sciences, Boltzmanngasse 3, 1090 Vienna, Austria}

\begin{abstract}
We consider the problem of detecting entanglement and nonlocality in one-dimensional (1D) infinite, translation-invariant (TI) systems when just near-neighbor information is available. This issue is deeper than one might think a priori, since, as we show, there exist instances of local separable states (classical boxes) which only admit entangled (non-classical) TI extensions. We provide a simple characterization of the set of local states of multi-separable TI spin chains and construct a family of linear witnesses which can detect entanglement in infinite TI states from the nearest-neighbor reduced density matrix. Similarly, we prove that the set of classical TI boxes forms a polytope and devise a general procedure to generate all Bell inequalities which characterize it. Using an algorithm based on matrix product states, we show how some of them can be violated by distant parties conducting identical measurements on an infinite TI quantum state. All our results can be easily adapted to detect entanglement and nonlocality in large (finite, not TI) 1D condensed matter systems.
\end{abstract}


\maketitle


\section{Introduction}

Imagine a scenario where a number of scientists are sent on a space exploration mission. Confined to separate vessels, they can only probe their immediate surroundings and communicate the outcomes of their experiments. We do not need to specify the exact nature of those experiments, but one could think, for instance, that each scientist is locally interacting with the vacuum state of a global quantum field. 

After conducting such experiments in different places what the scientists find out is that they always obtain the same statistics, no matter where they are, as long as the relative position between their vessels is the same. Unable to explore the whole universe, they postulate that this property must hold elsewhere, in addition to the regions they already visited. To model this assumption physically, we picture these scientists probing different sites of an infinite translation-invariant (TI) system. 

For further elucidation, we consider the simplest such scenario where the scientists live in a world that has one spatial dimension, so experiments are conducted at equidistant points on a straight line. The question we want to address is: from the information gathered by a small neighborhood of scientists, what global properties can they infer about the whole---infinite, unexplored---one dimensional (1D) TI system? In this paper, we will focus on two: (i) entanglement (namely, whether the local quantum state describing the neighborhood is incompatible with an underlying multiseparable state for the whole system) and (ii) Bell nonlocality (namely, whether it is impossible to simulate the statistics of the whole system with a classical device). 

Entanglement and nonlocality are two hallmark features of our world which signify a clear departure from classical physics~\cite{RevModPhys.81.865,RevModPhys.86.419}. The problem of certifying whether a given state is entangled and/or nonlocal is important in order to determine the type of correlations that are furnished by the state or to characterize the state as a useful resource for various quantum processing tasks. Most research in the entanglement of TI quantum systems has been focused on the entanglement between two distant sites~\cite{wootters_chains,PhysRevLett.92.087903,Osterloh2002rq}, or between a region of the chain and the rest of it~\cite{RevModPhys.82.277}. The multiseparability of quantum spin chains has been studied in eg.~\cite{PhysRevLett.103.100502,PhysRevLett.106.020401,Hauke}, and the permutation-invariant systems studied in~\cite{Tura13062014,Tura2015370}, when placed on a line, can be seen as translation-invariant as well. Unfortunately, the certification of entanglement or nonlocality in the aforementioned works requires the knowledge of correlations between arbitrarily distant sites, impossible to acquire in the \emph{gedankenexperiment} described above. Prior works on the non-classicality of infinite translation-invariant systems have focused on how to detect Bell nonlocality directly, i.e., by showing that the probed regions cannot be described classically~\cite{nonviolation,PhysRevA.73.022303}. The problem of global nonlocality detection in 1D TI systems via local measurements has been studied for finite number of parties~\cite{jordi_ti,jordi_ti2}. 

In this paper, we study the problem of certifying entanglement and nonlocality in 1D infinite TI systems from the information available to a finite number of parties exploring the chain. As we will show, our results also apply to the verification of entanglement and nonlocality in large (finite but not TI) 1D quantum many-body systems, so they may be particularly relevant for condensed matter experiments. In this regard, since all our entanglement witnesses and Bell inequalities only depend on near-neighbor two-body correlators, the quantum states maximally minimizing them can be prepared by cooling a condensed matter system described by a local TI Hamiltonian \footnote{In contrast to \emph{Bell local}, which is also called \emph{classical} and will be defined later, the word \emph{local} here means each term in the Hamiltonian only acts on a small neighborhood of a given site.}.

\section{Conceptual setup}
Consider infinitely many sites distributed equidistantly along a line. Any number of consecutive sites of the chain, say $1,...,n$, is described by a \emph{state} $\omega_{1,\ldots,n}$. Depending on the level of our description, such a state will correspond to (i) a quantum state $\rho_{1,2,...,n}$, or (ii) a conditional probability distribution (also called a \emph{box}) $P_{1,2,...,n}(a_1,a_{2},...,a_n|x_1,x_{2},...,x_n)$ for the values $a_1,a_{2},...,a_n$ of the local properties $x_1,x_{2},...,x_n$ at sites $1,2,...,n$, satisfying the non-signaling condition~\cite{popescu1994quantum}. To model separability and locality we will further require a coarser level of description: (iii) a probability distribution $P_{1,2,\ldots,n}(a_1,a_{2},...,a_n)$ for the values $a_1,a_{2},...,a_n$ of a local system property at sites $1,2,...,n$ respectively.


The quantum state $\rho$ of an infinite chain is \emph{multiseparable} if it can be decomposed in the form $\rho_{1,\ldots,n}=\int d\vec{\varrho} P(\varrho_{1},\ldots,\varrho_n)\varrho_1\otimes\ldots\otimes\varrho_n$,
for all $n$, where $P(\varrho_{1},\ldots,\varrho_n)$ is a probability density and $\varrho_{1},\ldots,\varrho_{n}$ are single-site density matrices. We will refer to $P(\varrho_{1},\ldots,\varrho_n)$ as a \emph{separable decomposition} for the state $\rho_{1,\ldots,n}$. Alternatively, the set of multiseparable states is the set of all states which can be generated via quantum one-site operations and classical communication, i.e., without the need of making the subsystems interact. 

Analogously, the box $P$ of an infinite chain is \emph{local} or \emph{classical}---that is, it does not violate any Bell inequality---if, for any $n$, the box $P_{1,2,...,n}(a_1,a_2,...,a_n|x_1,x_2,...,x_n)$ admits a local hidden variable model \cite{RevModPhys.86.419}. Namely, if there exists a probability distribution $\mu(\lambda)$ over a hidden variable $\lambda$ such that $P_{1,...,n}(a_1,a_2,...,a_n|x_1,x_2,...,x_n)=\sum_{\lambda}\mu(\lambda)Q_1(a_1|x_1,\lambda)Q_2(a_2|x_2,\lambda)...Q_n(a_n|x_n,\lambda)$, where $Q_k(a_k|x_k,\lambda)$ is a probability distribution for outcome $a_k$ at site $k$. Intuitively, the set of local boxes is the set of all black boxes which can be simulated via classical devices.

In this paper, we will be interested mostly in translation invariant (TI) states. An \emph{infinite TI state} $\Omega$ for the whole chain is defined as an infinite \emph{sequence} of states $(\Omega_{1,2,\ldots,s})_s$ satisfying $\Omega_{k,\ldots,k+m}=\Omega_{k+1,\ldots,k+m+1}$  for all $m,k$. 

Our goal is, given access to $\Omega_{1,2,...,r}$, to determine global properties of the infinite TI state $\Omega$, such as its entanglement (when $\Omega$ is a quantum state) or its nonlocality (when $\Omega$ corresponds to a box). However, our results can also be applied to finite and non-TI systems by means of the following -widely known- \emph{symmetrization procedure}, that allows us to construct an infinite TI state $\Omega$, given an $n$-site state $\omega_{1,2,...,n}$.

Consider a state $\Gamma$ of an infinite chain composed by infinitely many copies of $\omega$, i.e., $\Gamma \equiv \omega_{1,...,n} \otimes \omega_{1,...,n} \otimes \ldots$. The symbol $\otimes$ denotes the composition law of the state under consideration (tensor product for Hilbert spaces, multiplication for probability distributions, etc.). The state $\Gamma$ is clearly invariant under translations of $n$ sites. We can construct an infinite TI state $\Omega$ by summing $\Gamma$ with states obtained by translating $\Gamma$ by $k\in\{1,2,\ldots,n\}$ sites, each time with probability $\frac{1}{n}$, see Fig.~\ref{fig:symm}. State $\Omega$ is called a \emph{symmetrization} of $\omega_{1,2,\ldots,n}$, and the marginal $\Omega_{1,2,\ldots,r}$ of $r\leq n$ sites is given by
\begin{align}
\Omega_{1,\ldots,r}=\frac{1}{n}\left(\sum_{k=1}^{n-r+1}\omega_{k,\ldots,r+k-1}
+\sum_{k=1}^{r-1}\omega_{n-r+k+1,\ldots,n}\otimes \omega_{1,\ldots,k}\right).\label{symm_proc}
\end{align}
Many experimental setups in condensed matter physics do not allow the experimenter to probe each individual site of an $n$-site spin chain. Instead, one can estimate, via neutron diffraction, the average 2-site correlators or \emph{static structure factors} \cite{marshall1971theory}:
\be
\tilde{\omega}^{(n)}_{\left[r\right]}\equiv\frac{1}{n-r+1}\sum_{k=1}^{n-r+1} \omega_{k,k+r-1}.
\label{structure}
\ee
\noindent Given a large chain with structure factors $\{\omega^{(n)}_{\left[r\right]}\}_r$, the symmetrization procedure (\ref{symm_proc}) hence implies that there exists a TI state $\Omega$ with the property:
\begin{align}
\Omega_{1,r}=\omega^{(n)}_{\left[r\right]}+O\left(\frac{r}{n}\right).
\end{align}
\noindent Moreover, if $\omega_{1,2,\ldots,n}$ is a separable quantum state or a local or quantum box, then so is $\Omega$. It follows that, if the structure factors of the system violate an entanglement witness or Bell inequality for TI systems by an amount greater than $O(r/n)$, then the $n$-site chain must be, respectively, entangled or nonlocal. This means that, as long as we restrict ourselves to devising two-body entanglement witnesses and Bell inequalities, the conclusions which we will extract regarding TI systems also apply to large 1D condensed matter systems. Moreover, the witnesses constructed this way will be experimentally friendly since by construction they are maximally violated by the ground states of local TI Hamiltonians.

\begin{figure}
  \centering
  \includegraphics[width=8.5 cm]{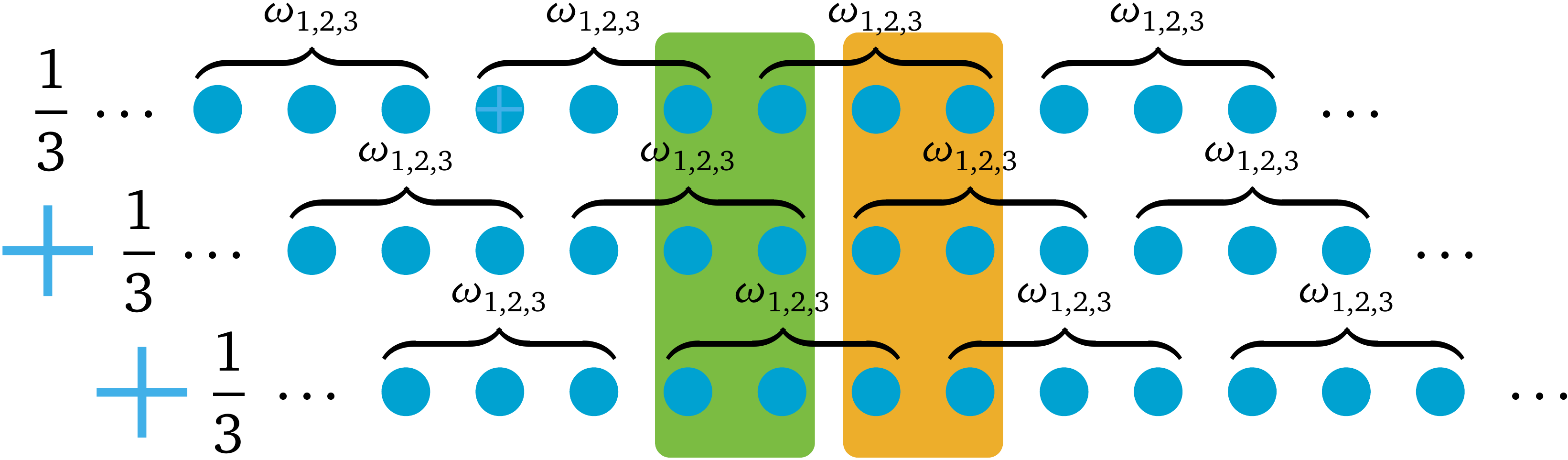}
  \caption{\textbf{The symmetrized state for $n=3$.} The green and yellow rectangles highlight the partial terms that contribute to two neighbouring 2-site reduced states respectively. These are seen to be equal (since they are sum of the same three partial terms).}
  \label{fig:symm}
\end{figure}

With the required notation and tools in place, we now turn to addressing our two main goals, namely, the detection of entanglement and nonlocality in large 1D chains by only using local information.

\section{Entanglement detection in large 1D chains}
Given a partial quantum state $\rho_{1,...,r}$, obtained by ignoring all but $r$ consecutive sites of an infinite TI quantum state, how can we ascertain whether the total state $\rho_{-\infty,...,\infty}$ is entangled? Of course, if $\rho_{1,...,r}$ itself is entangled, e.g. if it is not positive under partial transposition (PPT)~\cite{PhysRevLett.77.1413}, then nothing needs to be done. On the other hand, it is possible, as we illustrate below, that the total quantum state is \textit{entangled} even when $\rho_{1,...,r}$ is multi-separable. Given just access to $\rho_{1,...,r}$, the only question we can hope to answer is whether there exists a total multiseparable TI state from which the given partial state can be obtained by ignoring sites, i.e., whether $\rho_{1,...,r}$ admits an \emph{TI and separable (TIS) extension}.

Before addressing this problem, let us consider a related one: given an $r$-site probability distribution $P_{1,...,r}(x_1,...,x_r)$, decide whether it can be realized as the marginal of an infinite TI distribution $Q$. The solution of this problem is known for some time~\cite{frustration,Schlijper1985,Pivato,Goldstein2017} and remains part of the folklore of TI systems: $P_{1,...,r}(x_1,...,x_r)$ admits a TI extension if and only if
\be
P_{1,\ldots,r-1}(x_1,\ldots,x_{r-1})=P_{2,\ldots,r}(x_1,\ldots,x_{r-1}).
\label{consist}
\ee
 To see why this is true, consider the conditional probability distribution $P(x_n|x_1,\ldots,x_{r-1})\equiv \frac{P_{1,\ldots,r}(x_1,\ldots,x_r)}{P_{1,\ldots,r-1}(x_1,\ldots,x_{r-1})}$ (if the denominator is $0$, then any distribution is allowed). We can recursively extend the probability distribution $P_{1,...,r}(x_1,...,x_r)$ to a sequence $(Q_{1,\ldots,s})_{s}$ of probability distributions for increasingly larger chains $s\geq r$ by means of the recurrence relation:
\begin{align}\label{recursion}
Q(x_1,...,x_{s+1})&=Q(x_1,...,x_{s})P(x_{s+1}|x_{s-r+2},\ldots,x_{s}).
\end{align}
It is readily checked that $Q_{1,2,\ldots,r}=P_{1,2,\ldots,r}$, namely, $P_{1,2,\ldots,r}$ is a marginal of $Q_{1,2,\ldots,s+1}$ for $s\geq r$. From \eref{recursion}, it also follows that $Q_{1,2,\ldots,s+1}$ is a TI sequence provided that \eref{consist} holds. A characterization of the extreme points of the set ${\cal T}_r$ of $r$-site TI marginals, i.e., those TI marginals which cannot be expressed as convex combinations of other marginals, can be found in Appendix \ref{app:dominoes}.

Using the solution of the classical TI marginal problem, we will next derive a characterization of the set of states admitting a TIS extension. Assume that $\rho_{1,...,r}$ does indeed admit a TIS extension $\rho$, and let $P(\varrho_1,...,\varrho_n)$ define a separable decomposition for the state $\rho_{1,...,n}$. Applying the symmetrization procedure to $P(\varrho_1,...,\varrho_n)$ we obtain a TI distribution that we can regard as the separable decomposition of a chain state $\bar{\rho}$, whose reduced state $\bar{\rho}_{1,\ldots,r}$ is $O(r/n)$-close to $\rho_{1,\ldots,r}$. Since $n$ was arbitrary, we conclude that, if $\rho_{1,...,r}$ admits a TIS extension, then we can take its separable decomposition to be TI. Invoking Eq.~\ref{consist}, we thus have that an $r$-site quantum state $\rho_{1,...,r}$ admits a TIS extension iff it satisfies
\be
\rho_{1,\ldots,r}=\int d\vec{\varrho} P(\varrho_1,\ldots,\varrho_r)\varrho_1\otimes\ldots\otimes\varrho_r,
\label{sep_decomp}
\ee
with $P_{1,...,r-1}(\varrho_1,\ldots,\varrho_{r-1})=P_{2,...,r}(\varrho_1,\ldots,\varrho_{r-1})$.

Unfortunately, this characterization of TI separability is not very practical to detect entanglement. Indeed, given the state $\rho_{1,..,r}$, how to argue that it does \emph{not} admit a decomposition of the form in Eq.~\ref{sep_decomp}? This motivates us to look for simpler criteria to decide the existence of TIS extensions.

As a first attempt, we can apply the intuition from the characterization of TI probability distributions. Notice that if the state $\rho_{1,\ldots,r}$ has a TIS extension, then it must be separable and satisfy
\begin{align}
\rho_{1,\ldots,r-1}=\rho_{2,\ldots,r}.\label{state_consistency}
\end{align}
Are these conditions also sufficient to guarantee the existence of a TIS extension?

Let $\{\sigma_i\}_{i=x,z,y}$ denote the Pauli matrices. Using the Jordan-Wigner transformation~\cite{JordanWigner1928}, it can be shown that $\tr (\rho_{1,2}\sigma_y\otimes\sigma_x)\leq \frac{2}{\pi}$ for TI states $\rho$ (see Appendix \ref{rho_XY} for the proof). Now, the separable state $\varrho=\frac{1}{2}(\proj{+i}\otimes\proj{+}+\proj{-i}\otimes\proj{-})$, with $\sigma_x\ket{\pm}=\pm\ket{\pm}$, $\sigma_y\ket{\pm i}=\pm\ket{\pm i}$, satisfies $\varrho_1=\varrho_2=\frac{1}{2}\id$, but $\tr(\varrho \sigma_y\otimes\sigma_x)=1$. Thus separability plus condition Eq.~\ref{state_consistency} do not even guarantee the existence of a TI extension, separable or not.

This last observation, however, suggests a stronger criterion for the existence of a TIS extension, namely, to demand the state $\rho_{1,\ldots,r}$ to be both separable and the reduced state of an infinite TI state. Unfortunately, this criterion, although necessary, is still not sufficient to guarantee a TIS extension. To construct a counter example, we will first give two states, one TI and the other TIS, which can be seen as optimally witnessing translation-invariance and translation-invariance plus multiseparability.

First, in Appendix \ref{rho_XY}, we identify a TI state $\rho^1$ that saturates the inequality $\tr (\rho_{1,2}\sigma_y\otimes\sigma_x)\leq \frac{2}{\pi}$, with $\rho^1_{1,2}=\frac{1}{4}\id_4+\frac{1}{2\pi}(\sigma_y\otimes\sigma_x+\sigma_x\otimes\sigma_y)+\frac{1}{\pi^2}\sigma_z^{\otimes 2}$.

Second, in Appendix~\ref{witnesses} it is shown that all states $\rho_{1,2}\in \mathcal{B}(\C^2\otimes\C^2)$ with a TIS extension satisfy
\be
\tr(\rho_{1,2} \sum_{i,j=1}^3T_{ij}\sigma_i\otimes\sigma_j)\leq \frac{1}{2}\max_{\theta\in [0,2\pi]}\|e^{i\theta}T+e^{-i\theta}T^\dagger\|,
\ee
where $T_{ij}\in \mathbb{R}$. Taking $T_{i,j}=\delta_{i,2}\delta_{j,1}$, where $\delta$ is the Kronecker delta, implies that all states which have TIS extensions satisfy $\tr(\rho_{1,2}\sigma_y\otimes\sigma_x)\leq \frac{1}{2}$. This bound is tight, since it can be saturated by the TIS state $\rho^0\equiv\frac{1}{3}\sum_{s=1}^3\varrho^{s+1}\otimes \varrho^{s}$, where $\varrho^{1}=\varrho^{4}$, and  $\varrho^{1}$, $\varrho^{2}$, $\varrho^{3}$ are described, respectively, by the Bloch vectors $\frac{1}{\sqrt{2}}(1,1,0)$, $\frac{1}{\sqrt{2}}(-1,1,0)$, $\frac{1}{\sqrt{2}}(1,-1,0)$ \footnote{To see that this state admits a TIS extension, simply prepare the state $(\varrho^{1}\otimes\varrho^{1}\otimes\varrho^{1})^{\otimes \infty}$, invariant under translations by $3$ sites, and then subject it to a random translation $t=0,1,2$ with probability $1/3$. The resulting TIS state has the two-reduced density matrix $\rho^1$.}.

Now, consider the family of TI states $\rho^{\lambda}\equiv\lambda\rho^1+(1-\lambda)\rho^0$. Clearly, for $\lambda\in (0,1]$, all those states violate the entanglement witness $\langle\sigma_y\otimes\sigma_x\rangle\leq \frac{1}{2}$. Also, it can be verified that $\rho^{\lambda}_{1,2}$ is PPT for $\lambda\leq  \frac{2\pi^2}{12+12\pi-\pi^2}\approx 0.4956$. It follows that, for $\lambda\in (0, 0.4956]$, the states $\rho_{1,2}^\lambda$ are separable~\cite{Horodecki19961} and TI, but all their TI extensions are entangled. Similar effects have been reported in~\cite{tripartite,PhysRevA.93.020104} where the authors construct near-neighbor separable states (local boxes) which only admit entangled (nonlocal) global extensions. 

Thus, even though $\rho^\lambda_{1,2}$ is not entangled, its two-body correlators tell us that there exists a finite system size $n$ such that $\rho^\lambda_{1,...,n}$ is. 
This raises another interesting question, namely, how large $n$ must be. Consider a witness of the form $\tr(W\rho_{1,2})\leq S$ and suppose that $\rho_{1,2}$ violates it by an amount $\Delta>0$, i.e., $\tr(\rho_{1,2}W)=S+\Delta$. If there exists a TI extension $\rho$ of $\rho_{1,2}$ such that $\rho_{1,...,n}$ is separable, then applying the symmetrization procedure in Eq.~\ref{symm_proc} to $\rho_{1,...,n}$ would produce a separable TI state $\tilde{\rho}$ with $\tilde{\rho}_{1,2}=\frac{n-1}{n}\rho_{1,2}+\frac{1}{n}\rho_1\otimes\rho_1$. Since $\tilde{\rho}$ is separable and TI, it must satisfy $\tr(\tilde{\rho}_{1,2}W)\leq S$, from which it follows that $n\leq \frac{S-\tr(W\rho_1^{\otimes 2})}{\Delta}+1$. 

Therefore, contrary to the ordinary entanglement detection setup, the degree of violation of a linear entanglement witness has a clear operational meaning in the TI scenario thanks to an intrinsic notion of size: its inverse is proportional to the number $n$ of consecutive sites which $n$ parties must share in order to hold an entangled resource. This quantitative relation between nonseparability and size can be seen to hold for arbitrary entanglement witnesses, not just bipartite ones. It also extends to the realm of Bell nonlocality, that we will study next.

\section{Detecting nonlocality in large 1D chains}
Before tackling the characterization of nonlocality in TI systems, we will argue that certain infinite TI quantum systems are indeed non-classical. Take any bipartite quantum state $\rho\in \mathcal{B}(\H^{\otimes 2})$ which allows two parties to violate a Bell inequality $B$, and consider an infinite chain where each site $k$ holds two systems with Hilbert spaces $\H^{(k)}_1,\H^{(k)}_2$, with $\mbox{dim}(\H^{(k)}_1)=\mbox{dim}(\H^{(k)}_2)=\mbox{dim}(\H)$. If we distribute a copy of $\rho \in \mathcal{B}(\H^{(k)}_1 \otimes \H^{(k+1)}_2)$ to all neighbouring pairs ($k$, $k+1$), we end up with a TI chain configuration with the property that any pair of nearest neighbors can violate $B$. 

In this construction, the non-classicality of the whole chain is established by proving that the probed sites $1,2$ do not admit a local hidden variable model. Is this necessarily the case, or are there situations where the probed sites are classical, but nonetheless incompatible with an infinite classical TI box? We will need a complete characterization of nonlocality in 1D TI systems in order to answer this question.

Assume that the data available is of the form $P_{1,...,r}(a_1,...,a_r|x_1,...,x_r)$, where $a_k\in\{1,\ldots,d\}$ and $x_k\in\{0,\ldots,m-1\}$, with the promise that it arises from an infinite TI box. The task is to decide whether there exists a TI box $P$, compatible with the experimental data, and such that $P_{1,\ldots,n}(a_1,\ldots,a_n|x_1,\ldots,x_n)$ admits a local hidden variable model for all $n$. By Fine's theorem ~\cite{PhysRevLett.48.291}, the existence of a local hidden variable model for $P$ is equivalent to the existence of a global probability distribution $Q(\vec{a}_1,...,\vec{a}_\infty)$, with $\vec{a}_k\in\{1,\ldots,d\}^m$ such that 

\be
Q(a^{x_1}_1=b_1, ..., a^{x_n}_n=b_n)=P_{1,...,n}(b_1,...,b_n|x_1,...,x_n),
\ee
\noindent for all $n$. As in the characterization of TIS, we apply the symmetrization procedure over $Q_{1,...,n}$ in the limit $n\to\infty$ and find that we can assume the global distribution $Q$ to be TI.

As a vector of probabilities, the distribution $P_{1,...,r}(a_1,...,a_r|x_1,...,x_r)$ is a linear function $L$ of $Q(\vec{a}_1,\ldots,\vec{a}_r)$, whose only constraint is that it is the marginal of a TI distribution. Since this is equivalent to satisfying Eq.~\ref{consist}, it follows that we can characterize $P_{1,...,r}(a_1,...,a_r|x_1,...,x_r)$ via linear programming \cite{LP}. 

%

The set of all marginal distributions $P_{1,...,r}(a_1,...,a_r|x_1,...,x_r)$ arising from a 1D classical TI chain thus forms a \emph{convex polytope}, i.e., the convex hull of a finite number of vertices. This in itself is a very surprising result:  due to the presence of infinitely many parties, there is no a priori reason to expect this set to be a polytope. Actually, in the 2D case, the boundary of the corresponding set has both flat and smoothly curved parts and does not admit an exact computational characterization \cite{2dti}.

Each polytope has a dual description in terms of a finite set of linear inequalities or \emph{facets}. The transformation between these two descriptions can be done algorithmically albeit generally with high complexity. Using the software \emph{PANDA}~\cite{panda}, we enumerated the facets of the classical polytope describing two-input/two-output nearest-neighbor and next-to-nearest-neighbor distributions (that is, $P_{1,2}(a_1,a_2|x_1,x_2)$, $P_{1,3}(a_1,a_3|x_1,x_3)$, with $x_k\in\{0,1\}$, $a_k\in\{1,2\}$). The polytope has 32372 facets, which reduce to 2102 inequivalent inequalities after taking one-site relabellings and the reflection of the chain into consideration. Two examples are given by
\begin{align}
I_{\text{T}}&\equiv-2E_0-4E_1-2E^{1,2}_{00}+2E^{1,2}_{01}+2E^{1,2}_{10}+2E^{1,2}_{11}\nonumber\\
&+E^{1,3}_{00}+E^{1,3}_{11}\geq -4,\label{ineq1195}\\
I_{\text{G}}&\equiv -4E_0-6E_1-3E^{1,2}_{00}+2E^{1,2}_{01}+3E^{1,2}_{10}+2E^{1,2}_{11}\nonumber\\
&+2E^{1,3}_{00}+E^{1,3}_{10}+E^{1,3}_{11}\geq -6,\label{ineq3621}
\end{align}
where $E_{x}\equiv \langle A^1_x\rangle=\sum_{a=0,1}P_{1}(a|x)(-1)^a$ and $E_{xy}^{i,j}\equiv \langle A^i_x A^j_y\rangle=\sum_{a,b=0,1}P_{i,j}(a,b|x,y)(-1)^a(-1)^b$. Here $A^i_x$ denotes the observable corresponding to measuring property $x$ at site $i$ and assigning it the numerical value $(-1)^a$.

In order to estimate the quantum value of an inequality given by $I\equiv \sum_{x,y=0,1}\frac{1}{2}C_x\cdot E_x+C_{xy}^{AB}\cdot E_{xy}^{1,2}+C_{xy}^{AC}\cdot E_{xy}^{1,3}$, we associate a quantum system of dimension $d=4$ to each of the sites. $d=4$ is chosen because we could not violate any inequality by using quantum systems with lower dimensions on each site. We then identify the observables $A_0,\;A_1$ at each site with the operators $A_0 \equiv M(0,0),\; A_1 \equiv M(\theta,\phi)$ where
\begin{align}
M(\theta,\phi)&\equiv\begin{pmatrix}
 \cos (\theta ) & \sin (\theta ) & 0 & 0 \\
 \sin (\theta ) & -\cos (\theta ) & 0 & 0 \\
 0 & 0 & \cos (\phi ) & \sin (\phi ) \\
 0 & 0 & \sin (\phi ) & -\cos (\phi )
\end{pmatrix}.
\end{align}


This way, fixing $\theta,\phi$, we can map the original Bell inequality to the 3-local Hamiltonian
\begin{align}
H&\equiv \sum_{i=1}^{\infty}\sum_{x,y=0,1}\frac{1}{2}C_x\cdot A_x^i+C_{xy}^{AB}\cdot A_x^i\otimes A_y^{i+1}+C_{xy}^{AC}\cdot A_x^i\otimes A_y^{i+2}.
\end{align}

The minimum quantum value of the Bell inequality (under the corresponding measurement settings) corresponds to the ground state energy per site of this Hamiltonian. The computation of the latter was carried out over infinite Matrix Product States (iMPS) using a combination of the Time Evolving Block Decimation (TEBD) method~\cite{PhysRevLett.98.070201} and the tool \emph{Open Source MPS}~\cite{OSMPS}, which implements a variant of the Density Matrix Renormalization Group (DMRG) method~\cite{DMRG_Schollwoeck}. Using these tools, we find the violations $I_{\text{T}}=-4.1847$ with $\theta=0.077$, $\phi=1.874$ and $I_{\text{G}}=-6.1798$ with  $\theta=6.236$, $\phi=4.175$. More inequalities, including violations using DMRG and lower bounds on the nonsignaling and quantum values can be found in Appendix~\ref{app:ineqs}.

The violations obtained above may not be optimal and a see-saw like method, similar to~\cite{PhysRevA.82.022116}, can be used to enhance them. In such a method, the optimization is divided into two rounds: in one round the measurements are held fixed and the optimization is over the state, and in the other round the state is held fixed and the measurements are optimized. By repeating these two rounds, a see-saw method usually converges to better violations than naive methods such as the one used above. While optimization over states by fixing the measurements can be done using TEBD or DMRG, the optimization over measurements with a fixed state involves further complications because the measurement operators make the objective function bilinear. Fortunately, this bilinearity can be removed if in addition to the 4-dimensional quantum state, each party is given access to a classical TI register, as described by the protocol given in Appendix~\ref{app:variational}. Using this protocol, the violation of $I_{\text{G}}$ can be increased to $-6.1907$. At first glance, it may seem surprising that, by giving them access to shared randomness, the parties are able to increase their violation. Note, though, that the extreme points of TI distributions are not necessarily deterministic.

We verified that the TI local value of $I_{\text{T}}$ cannot be beaten by local tripartite boxes $P_{1,2,3}(a_1,a_2,a_3|x_1,x_2,x_3)$ with $P_{1,2}(a,a'|x,x')=P_{2,3}(a,a'|x,x')$. This means that a TI box $P$ can only violate \eref{ineq1195} when the tripartite box $P_{1,2,3}$ describing the state of three consecutive sites does not admit a classical model. In other words: $I_{\text{T}}$ is just detecting standard tripartite nonlocality. $I_{\text{G}}$, however, is different. While the tripartite box $Q_{1,2,3}$ generated by nearest and next-to-nearest neighbors of the state achieving the violation is not local, TI local noise can be added to $Q$ to turn it into a new TI box $\tilde{Q}$, with $\tilde{Q}_{1,2,3}$ tripartite local, while keeping a violation of $I_{\text{G}}(\tilde{Q})\approx -6.1525$. Similarly to the entanglement case, even though the behavior of the tripartite box $\tilde{Q}_{1,2,3}$ can be reproduced with classical devices, for some $n$ no local hidden variable model can possibly describe an $n$-site box with marginals $\tilde{Q}_{1,2,3}$.

\section{Conclusions}
In this paper we showed how to derive global properties of infinite 1D TI systems when only local information is available. We provided a characterization of the reduced density matrices of TI multiseparable states and used it to derive entanglement witnesses for infinite TI qubit chains. Along the way, we constructed examples of TI states with a separable nearest-neighbours density matrix which nonetheless only admit entangled TI extensions. Regarding nonlocality, we fully characterized the set of $r$-partite boxes obtained by probing the sites of a classical infinite TI chain. Similarly to the entanglement case, we identified a classical tripartite box which only admits non-classical TI extensions.

For future research, it would be interesting to develop effective methods to bound the nonlocality of TI quantum and nonsignalling systems. Also, it would be desirable to extend some of our results to higher spatial dimensions. 

\noindent\emph{Acknowledgements}
M.N. and Z.W. acknowledge the FQXi grant `Towards an almost quantum physical theory'. The authors would also like to thank Tam\'as V\'ertesi for stimulating discussions and Valentin Stauber, Daniel Jaschke, Michael Wall for their help with part of the MPS simulation.

\bibliography{../biblio}

\clearpage
\begin{appendix}
\section{The polytope ${\cal T}_n$ and its extreme points}\label{app:dominoes}
Having found the necessary and sufficient conditions for a classical distribution to have a TI extension, a natural question to ask is how to find the extreme points of the set $\T_n$ of distributions satisfying Eq.~\ref{consist}. It turns out that the answer to this question has a nice combinatorial flavor. Let $\{1,\ldots,d\}$ be the set of possible outcomes for each random variable $x_k$. A sequence of vectors $\{y^{(s)}\}_s\subset\{1,\ldots,d\}^n$ is called a \emph{domino line} if $(y^{(s)}_1,\ldots,y^{(s)}_{n-2})=(y^{(s+1)}_0,\ldots,y^{(s+1)}_{n-1})$ for all $s$ (see Fig.~\ref{fig:dominoes} for an example). A finite sequence of dominoes $(y^{(s)})_{s=1}^m$ is called a \emph{domino loop} if $(y^{(m)}_1,\ldots,y^{(m)}_{n-2})=(y^{(1)}_0,\ldots,y^{(1)}_{n-1})$. Moreover, the loop is \emph{irreducible} if any strict subset of it cannot be re-ordered to create another loop. By the pigeonhole principle, the maximum size of an irreducible domino loop is finite and satisfies $m\leq d^n$. The following theorem characterizes the extreme points in terms of the domino loops.

\begin{figure}[htbp!]
  \centering
  \includegraphics[width=8.5 cm]{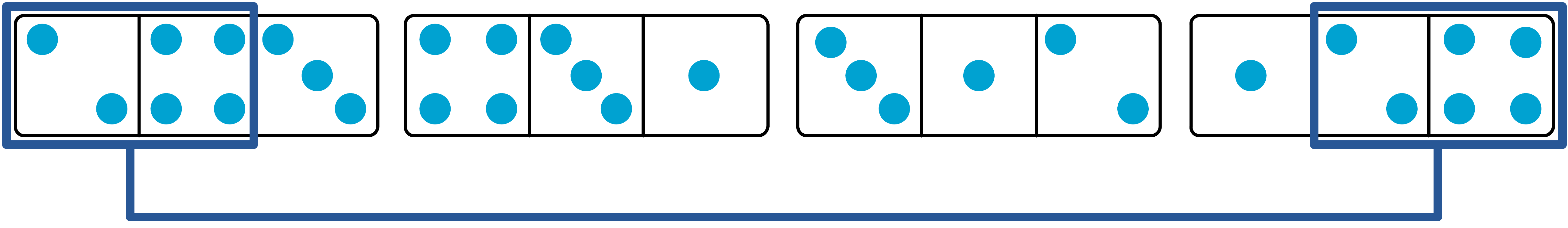}
  \caption{\textbf{An example of an irreducible domino loop.} Taking $n=3, d=6$, each vector $y$ can be depicted as a domino tile with three numbers. The last two numbers in each tile match the first two numbers of the next tile (the next tile of the last tile is the first one), making it a domino loop. This loop is irreducible, since each tile has a unique successor.}
  \label{fig:dominoes}
\end{figure}

\begin{theo}
A probability distribution $P_{\bar{y}}(\bar{x})$ is an extreme point of $\T_n$ iff it is of the form
\be
P_{\bar{y}}(\bar{x})=\frac{1}{m}\sum_{s=0}^{m-1}\delta\{\bar{x},y^{(s)}\},
\label{extreme}
\ee
where $\delta$ is the Kronecker delta and $(y^{(s)})_{s=0}^{m-1}$ is an irreducible domino loop.
\end{theo}

\begin{proof}
Let $P(x_1,\ldots,x_n)$ be an $n$-partite distribution satisfying Eq.~\ref{consist}, and let $(y_0,\ldots,y_{n-1})$ be such that $P(y_0,\ldots,y_{n-1})\not=0$. Since $P(y_0,\ldots,y_{n-1})\not=0$, it follows that $P_{2,\ldots,n}(y_1,\ldots,y_{n-1})\not=0$. Hence, by Eq.~\ref{consist}, $P_{1,\ldots,n-1}(y_1,\ldots,y_{n-1})\not=0$, and so there must exist $y_{n}\in\{1,\ldots,d\}$ such that $P(y_1,\ldots,y_n)\not=0$. Iterating this reasoning, we end up with a line of dominoes $y^{(0)}, y^{(1)},\ldots$ with the property that $P(y^{(s)})\not=0$ for every $s$. By the pigeon-hole principle, there must exist $k>j$ such that $y^{(k)}=y^{(j)}$, and so the dominoes $y^{(j)},y^{(j+1)},\ldots,y^{(k-1)}$ constitute a domino loop. Let $z^{(1)},\ldots,z^{(m)}$ be any irreducible domino loop contained in the last set. Then, one can verify that

\be
Q(\bar{x})\equiv\frac{1}{m}\sum_{s=0}^{m-1}\delta\{\bar{x},z^{(s)}\}
\ee
 
\noindent satisfies Eq.~\ref{consist}. Define $\lambda\equiv\min_s P(z^{(s)})$, and call $t\in\{0,\ldots,m-1\}$ the argument of the optimizer. Then, $P'(\bar{x})\equiv\frac{P(\bar{x})-m\lambda Q(\bar{x})}{1-m\lambda}$ is a probability distribution complying with Eq.~\ref{consist}, with $P'(z^{(t)})=0$, and such that $P(x)=m\lambda Q(x)+(1-m\lambda)P'(\bar{x})$. Applying this procedure again on $P'(\bar{x})$ and iterating, we obtain a sequence of probability distributions $P(X),P'(x),P''(x),\ldots$ with strictly decreasing support. This sequence must thus come to an end, and so we have proven that the original distribution $P(x)$ is a convex combination of distributions of the form (\ref{extreme}).

It just suffices to prove that the latter distributions are extreme, i.e, that they cannot be decomposed as convex combinations of other distributions satisfying Eq.~\ref{consist}. Suppose, then, that $Q(x)=pQ'(x)+(1-p)Q''(x)$, with $Q'(x)$ being an extreme point of the set ${\cal T}_n$. Using the same argument as before, we can show that there must exist a sequence of vectors $z^{(0)},\ldots,z^{(p-1)}$ with $Q'(z^{(s)})\not=0$ for $s=0,\ldots,p-1$, forming an irreducible domino loop. Such vectors obviously satisfy $Q(z^{(s)})\not=0$. Since the vectors constituting the support of $Q(x)$ are themselves an irreducible domino loop, and $z^{(0)},\ldots,z^{(p-1)}$ are a subset of them, then they must necessarily be the same. Thus, for some $1>\mu>0$, $Q'(x)=\mu Q(x)+(1-\mu)R(x)$, contradicting the hypothesis that $Q'(x)\not=Q(x)$ is an extreme point. 

\end{proof}

\section{Entanglement witnesses}
\label{witnesses}
We wish to find an upper bound on $W_T$, defined as:

\begin{align}
&W_T\equiv\max \tr(\rho_{1,2} \sum_{i,j=1}^3T_{ij}\sigma_i\otimes\sigma_j),\nonumber\\
&\text{s.t. }\rho \text{ is a TIS state}.
\end{align}

In principle, we would need to optimize the above functional over all states of the form $\rho_{1,2}=\int d\vec{\psi} P(\psi_1,\psi_2)\proj{\psi_1}\otimes \proj{\psi_2}$, with $P_1(\psi)=P_2(\psi)$. However, from the characterization of the extreme points of ${\cal T}_n$, we infer that it is enough to restrict to quantum states of the form $\frac{1}{m}\sum_{s=1}^m\proj{\psi_s}\otimes \proj{\psi_{s+1}}$, with $\psi_{m+1}=\psi_{1}$. Defining $\proj{\psi_s}=\frac{1}{2}(\id_2+\bar{v}_s\cdot\bar{\sigma})$, we have, by straightforward computation, that the above problem can be rephrased as

\begin{align}
&\max \frac{1}{m}\sum_{s=1}^m\bar{v}_s\cdot T\cdot \bar{v}_{s+1},\nonumber\\
&\mbox{s.t. }\|\bar{v}_s\|^2=1,\forall s.
\end{align}

\noindent Defining $\ket{\phi}=\frac{1}{\sqrt{m}}\sum_{s=1}^m\sum_{i=1}^3v_s^i\ket{s}\ket{i}$, it is immediate that $\braket{\phi}{\phi}=1$ and the objective function in the equation above equals to

\be
\frac{1}{2}\bra{\phi}(A\otimes T+A^\dagger\otimes T^\dagger)\ket{\phi},
\ee

\noindent where $A=\sum_{s=1}^m\ket{s}\bra{s+1}$ (we identify the states $\ket{m+1}$ and $\ket{1}$). Now, the eigenvectors of the operator between brackets can be written as $\ket{\tilde{k}}\ket{\phi_{k,j}}$, where $\ket{\tilde{k}}=\frac{1}{\sqrt{m}}\sum_{s=0}^{m-1}e^{\frac{i2\pi sk}{m}}\ket{s}$, and $\{\ket{\phi_{k,j}}\}_{j=1}^3$ are the eigenvectors of the operator $e^{\frac{i2\pi k}{m}}T+h.c.$. Substituting, and letting $m$ be arbitrarily large, we find that

\be
W_T\leq \frac{1}{2}\max_{\theta\in [0,2\pi]}\|e^{i\theta}T+e^{-i\theta}T^\dagger\|.
\ee

\section{Maximal violation of our entanglement witnesses by TI states}
\label{rho_XY}
The purpose of this appendix is to find the two-reduced density matrix $\rho_{1,2}$ of a TI state $\rho$ maximizing the energy of the Hamiltonian

\be
H_R=\frac{1}{n}\sum_{j=1}^n \sigma_y^{(j)}\sigma_x^{(j+1)},
\ee

\noindent assuming closed boundary conditions.

Let us assume $n=2m+1$. By using the Jordan-Wigner transformation~\cite{JordanWigner1928}, the spin operators can be written as

\begin{align}
&\sigma_z^{(k)}=[a_k,a^\dagger_k],\nonumber\\
&\sigma_x^{(k)}\sigma_x^{(k+1)}=(a_k^\dagger-a_k)(a_{k+1}^\dagger +a_{k+1}),\nonumber\\
&\sigma_y^{(k)}\sigma_y^{(k+1)}=-(a_k^\dagger+a_k)(a_{k+1}^\dagger -a_{k+1}),\nonumber\\
&\sigma_x^{(k)}\sigma_y^{(k+1)}=i(a_k^\dagger-a_k)(a_{k+1}^\dagger-a_{k+1}),\nonumber\\
&\sigma_y^{(k)}\sigma_x^{(k+1)}=i(a_k^\dagger+a_k)(a_{k+1}^\dagger+a_{k+1}),
\label{correspondence}
\end{align}

\noindent where the Hermitian variables $\{a_{k}: k=-m,...,m\}$ satisfy the fermionic canonical anticommutation relations:

\begin{align}
\{a_{k},a^\dagger_{k'}\}&=\delta_{kk'},\nonumber\\
\{a_{k},a_{k'}\}&=0.
\end{align}

Let us define the Fourier-transformed operators

\be
\tilde{a}_k=\frac{1}{\sqrt{n}}\sum_{j=-m}^{m}e^{\frac{i2\pi k j}{n}}a_j.
\ee

\noindent They obviously satisfy the canonical anticommutation relations, and the transformation can be inverted via the equation below:

\be
a_k=\frac{1}{\sqrt{n}}\sum_{j=-m}^{m}e^{\frac{-i2\pi k j}{n}}\tilde{a}_j.
\ee

In terms of the Fourier-transformed operators, $H_R$ admits a very simple expression (note that we are neglecting the $O(1/n)$ interaction between sites $-m$ and $m$):

\begin{align}
&H_R=i\frac{1}{n}\sum_{j=-m}^{m}(a_j+a_j^\dagger)(a_{j+1}+a_{j+1}^\dagger)=\nonumber\\
&i\frac{1}{n^2}\sum_{j=-m}^{m}\sum_{k,k'=-m}^{m}e^{\frac{-i2\pi k j}{n}}e^{\frac{i2\pi k' (j+1)}{n}}(\tilde{a}_k+\tilde{a}_{-k}^\dagger)(\tilde{a}_{-k'}+\tilde{a}^\dagger_{k'})=\nonumber\\
&i\frac{1}{n}\sum_{k=-m}^{m}e^{\frac{i2\pi k}{n}}(\tilde{a}_k+\tilde{a}_{-k}^\dagger)(\tilde{a}_{-k}+\tilde{a}_{k}^\dagger)=\nonumber\\
&i\frac{1}{n}(\tilde{a}_0+\tilde{a}_{0}^\dagger)(\tilde{a}_{0}+\tilde{a}_{0}^\dagger)+\nonumber\\
&+i\frac{1}{n}\sum_{k>0}^{m}e^{\frac{i2\pi k}{n}}(\tilde{a}_k+\tilde{a}_{-k}^\dagger)(\tilde{a}_{-k}+\tilde{a}_{k}^\dagger)+\nonumber\\
&e^{-\frac{i2\pi k}{n}}(\tilde{a}_{-k}+\tilde{a}_{k}^\dagger)(\tilde{a}_{k}+\tilde{a}_{-k}^\dagger)\approx\nonumber\\
&-\frac{1}{n}\sum_{k=1}^{m}2\sin\left(\frac{2\pi k}{n}\right)\{(\tilde{a}_k+\tilde{a}_{-k}^\dagger)(\tilde{a}_{-k}+\tilde{a}_{k}^\dagger)-1\}=\nonumber\\
&\frac{1}{n}\sum_{k=1}^{m}2\sin\left(\frac{2\pi k}{n}\right)(1-2b^\dagger_{k}b_{k})\equiv \tilde{H}_R.
\end{align}

\noindent In the last identity, we ignore the contribution of the $k=0$ mode. Here $b_k$ belong to another basis of fermionic operators, given by:

\begin{align}
b_k&=\frac{\tilde{a}^\dagger_k+\tilde{a}_{-k}}{\sqrt{2}},\, k=1,...,m,\nonumber\\
\bar{b}_k&=\frac{\tilde{a}_k-\tilde{a}^\dagger_{-k}}{\sqrt{2}},\, k=1,...,m,\nonumber\\
b_0&=a_0.
\end{align}

The maximum average value of the operator $\tilde{H}_R$ is hence 

\be
\frac{1}{2m+1}\sum_{k=1}^{m}2\sin\left(\frac{2\pi k}{2m+1}\right),
\ee
\noindent which, in the limit $n\to\infty$, converges to $2/\pi$. This value can be achieved by a state $\ket{\Phi_0}$ such that $b_k\ket{\Phi_0}=\bar{b}_k\ket{\Phi_0}=0$ for all $k$. 

Next, we will identify the nearest-neighbor density matrix of an infinite TI state $\rho$ achieving this value. First, define  

\begin{align}
H_{L}\equiv\frac{1}{n}\sum_{j=0}^{n}\sigma_{x}^{(j)}\sigma_{y}^{(j+1)}.
\end{align}

\noindent We observe that

\be
H_L\approx\tilde{H}_{L}=\frac{1}{n}\sum_{k=1}^{m}2\sin\left(\frac{2\pi k}{n}\right)(1-2\bar{b}^\dagger_{k}\bar{b}_{k}).
\ee

\noindent It follows that, in the limit $n\to\infty$, $\bra{\Phi_0}H_L\ket{\Phi_0}=2/\pi$, and so $\ket{\Phi_0},b^\dagger_0\ket{\Phi_0}$ are a basis for the states maximizing the energy of the Hamiltonian $\tilde{H}=\tilde{H}_R+\tilde{H}_L$. Our TI state $\rho$ will be the result of applying a symmetrization over the state $\tilde{\rho}=(\proj{\Phi_0}+b^\dagger_0\proj{\Phi_0}b_0)/2$ in the limit $n\to\infty$. It is easy to see that $\tr(\sigma_i\otimes\sigma_j\rho_{12})=\frac{1}{n}\sum_{k=1}^n \langle\sigma_i^{(k)}\sigma_j^{(k+1)}\rangle$, where the operator average is computed over $\tilde{\rho}$.

Now, $\tilde{H}$ is invariant under the action of the parity operator $\sigma_z^{\otimes n}$: the only non-zero components of the nearest-neighbors reduced density matrix of $\rho$ can thus be $\{\tr(\rho_{12}\sigma_{x,y}\otimes\sigma_{x,y})\}$, $\tr(\rho_{12}\sigma_z\otimes\id)=\tr(\rho_{12}\id\otimes\sigma_z)$ and $\tr(\rho_{12}\sigma_z\otimes\sigma_z)$.

In order to estimate $\tr(\rho_{12}\sigma_z\otimes\id)$, $\tr(\rho_{12}\sigma_z\otimes\sigma_z)$, we must compute the values $\frac{1}{n}\sum_{j}\langle a_j^\dagger a_{j}\rangle$, $\frac{1}{n}\sum_{j}\langle a_j^\dagger a_{j} a_{j+1}^\dagger a_{j+1}\rangle$. Let us carry out the calculation of the former one. We find that

\begin{align}
&\frac{1}{n}\sum_{j}\langle a_j^\dagger a_{j}\rangle=\frac{1}{n^2}\sum_{k,k'=-m}^m\sum_{j} e^{2\pi kj/n}e^{-2\pi k'j/n}\langle \tilde{a}_{k}^\dagger \tilde{a}_{k'}\rangle=\nonumber\\
&\frac{1}{n}\sum_{k=-m}^m\langle \tilde{a}_{k}^\dagger \tilde{a}_{k}\rangle\approx\nonumber\\
&\frac{1}{n}\sum_{k=1}^{m}\langle \tilde{a}_{k}^\dagger \tilde{a}_{k}\rangle+\langle \tilde{a}_{-k}^\dagger \tilde{a}_{-k}\rangle=\nonumber\\
&\frac{1}{n}\sum_{k=1}^{m}\left\langle\left(\frac{\bar{b}^\dagger_k+b_k}{\sqrt{2}}\right)\left(\frac{\bar{b}_k+b^\dagger_k}{\sqrt{2}}\right)\right\rangle+\nonumber\\
&+\left\langle\left(\frac{\bar{b}_k-b_k^\dagger}{\sqrt{2}}\right)\left(\frac{\bar{b}^\dagger_k-b_k}{\sqrt{2}}\right)\right\rangle=\frac{m}{n}.
\end{align}

\noindent For simplicity, in the third line we have neglected the $O(1/n)$ contribution of the $0$ mode. In the limit $n\to\infty$, we thus have that $\frac{1}{n}\sum_{j}\langle a_j^\dagger a_{j}\rangle=1/2$. 

The computation of the next average is more involved. We have that

\begin{align}
&\frac{1}{n}\sum_{j=-m}^m\langle a_j^\dagger a_{j} a_{j+1}^\dagger a_{j+1}\rangle=\nonumber\\
&\frac{1}{n^2}\sum_{k,k',k'',k'''}^m e^{\frac{i2\pi(k''-k''')}{n}}\delta(k-k'+k''-k''')\langle a_k^\dagger a_{k'} a_{k''}^\dagger a_{k'''}\rangle.
\end{align}

\noindent Given the structure of $\tilde{\rho}$, the averaged term on the second line can be non-zero only when $|k|,|k'|,|k''|,|k'''|$ are matched in pairs. This leaves us three options: (a) $|k|=|k'|$, $|k''|=|k'''|$; (b) $|k|=|k''|$, $|k'|=|k'''|$; (c) $|k|=|k'''|$, $|k'|=|k''|$. Taking into account the presence of the delta function in the integrand, the only contribution of case (a) that does not vanish in the limit $n\to\infty$ is $k=k'=v$, $k''=k'''=w$. Analogously, the only contributions which we need to take into account in the limit $n\to\infty$ for cases (b), (c) are, respectively, $k=-k''=v$, $k'=-k'''=w$ and $k=k'''=v$, $k'=k''=w$. We conclude that the expression above equals (modulo $O(1/n)$):

\begin{align}
&\frac{1}{n^2}\sum_{v,w}\langle\tilde{a}_v^\dagger\tilde{a}_v\tilde{a}_w^\dagger\tilde{a}_w\rangle-e^{\frac{i2\pi (w-v)}{n}}\langle\tilde{a}_v^\dagger\tilde{a}^\dagger_{-v}\tilde{a}_w\tilde{a}_{-w}\rangle+\nonumber\\
&e^{\frac{i2\pi (v-w)}{n}}\langle \tilde{a}_v^\dagger\tilde{a}_{v}\tilde{a}_{w}\tilde{a}^\dagger_w\rangle.
\end{align}

Now, on one hand it is easy to verify that $\langle \tilde{a}^\dagger_v\tilde{a}_v\rangle=\langle \tilde{a}_v\tilde{a}^\dagger_v\rangle=\frac{1}{2}$, $\langle \tilde{a}^\dagger_v\tilde{a}^\dagger_{-v}\rangle=-\langle \tilde{a}_v\tilde{a}_{-v}\rangle=\frac{1}{2}\mbox{sign}(v)$. On the other hand, for any two monomials of degree two $f,g$, $\langle f(\tilde{a}_v,\tilde{a}^\dagger_v)g(\tilde{a}_w,\tilde{a}^\dagger_w)\rangle=\langle f(\tilde{a}_v,\tilde{a}^\dagger_v)\rangle\langle g(\tilde{a}_w,\tilde{a}^\dagger_w)\rangle$ for $|v|=|w|$. Since values $v,w$ satisfying $|v|\not=|w|$ give a contribution that vanishes in the limit of large $n$, we conclude that

\begin{align}
&\lim_{n\to\infty}\frac{1}{n}\sum_{j=-m}^m\langle a_j^\dagger a_{j} a_{j+1}^\dagger a_{j+1}\rangle=\frac{1}{4}+\nonumber\\
&\frac{1}{(4\pi)^2}\int_{-\pi}^{\pi} dx \int_{-\pi}^{\pi} dy e^{i(x-y)}\mbox{sgn}(x)\mbox{sgn}(y)+e^{i(y-x)}=\nonumber\\
&\frac{1}{4}+\frac{1}{\pi^2}.
\end{align}

The above implies $\tr(\rho_{12}\sigma_z\otimes\id)=0$, $\tr(\rho_{12}\sigma_z\otimes\sigma_z)=4/\pi^2$. Similarly to the previous cases, by expressing $\sum_{k}\sigma_x^{(k)}\sigma_x^{(k+1)}$ and $\sum_{k}\sigma_y^{(k)}\sigma_y^{(k+1)}$ in terms of $b_k,\bar{b}_k$, we conclude that $\tr(\rho_{12}\sigma_x\otimes\sigma_x)=\tr(\rho_{12}\sigma_y\otimes\sigma_y)=0$.

Thus, 

\be
\rho_{1,2}=\frac{1}{4}\id_4+\frac{1}{2\pi}(\sigma_y\otimes\sigma_x+\sigma_x\otimes\sigma_y)+\frac{1}{\pi^2}\sigma_z^{\otimes 2}.
\ee

\section{Inequalities for nearest- and next-to-nearest neighbor correlators of infinite TI systems}
\label{app:ineqs}
Table.\ref{ineqs_list} lists the inequalities with nearest- and next-to-nearest neighbor correlators which can be violated using infinite TI quantum systems. Column $\mathcal{L}$ gives the local bounds for each inequality. $I_{\text{T}}$ in the main text is number 2 in the table while $I_{\text{G}}$ is number 4. Table.~\ref{ineqs_violations} gives the quantum value given by DMRG ($\mathcal{Q}$), the lower bound on the nonsignaling value ($\inf{\mathcal{NS}}$), the lower bound on the quantum value ($\inf{\mathcal{Q}}$) and whether the inequality can detect genuine TI nonlocality, as defined in the main text (\emph{Genuine}). These tools will be explained in more detail in a forthcoming paper~\cite{InPrep}.

The quantum values are obtained by using the software \emph{Open Source MPS}~\cite{OSMPS}. An inequality with coefficients $\{C_0,C_1,C_{00}^{AB},C_{01}^{AB},C_{10}^{AB},C_{11}^{AB},C_{00}^{AC},C_{01}^{AC},C_{10}^{AC},C_{11}^{AC}\}$ is turned into a 3-site Hamiltonian
\begin{align}
H\equiv &\sum_{i=1}^{\infty}(C_0\cdot A_0^i+C_1\cdot A_1^i+C_{00}^{AB}\cdot A_0^i\otimes A_0^{i+1}+C_{01}^{AB}\cdot A_0^i\otimes A_1^{i+1}\nonumber\\
&+C_{10}^{AB}\cdot A_1^i\otimes A_0^{i+1}+C_{11}^{AB}\cdot A_1^i\otimes A_1^{i+1}+C_{00}^{AC}\cdot A_0^i\otimes A_0^{i+2}\nonumber\\
&+C_{01}^{AC}1\cdot A_0^i\otimes A_1^{i+2}+C_{10}^{AC}\cdot A_1^i\otimes A_0^{i+2}+C_{11}^{AC}\cdot A_1^i\otimes A_1^{i+2}),
\end{align}
with
\begin{align}
A_0&=\begin{pmatrix}
1&0&0&0\\
0&-1&0&0\\
0&0&1&0\\
0&0&0&-1
\end{pmatrix},\\
A_1&=\begin{pmatrix}
 \cos (\theta ) & \sin (\theta ) & 0 & 0 \\
 \sin (\theta ) & -\cos (\theta ) & 0 & 0 \\
 0 & 0 & \cos (\phi ) & \sin (\phi ) \\
 0 & 0 & \sin (\phi ) & -\cos (\phi )
\end{pmatrix}.
\end{align}
The quantum value then corresponds to the ground state energy per site of this Hamiltonian. The values for $\theta$ and $\phi$ for each inequality are given in the table.

\begin{table}[htbp!]
\begin{center}
\begin{tabular}{|c|c||c|c|c|c|c|c|c|c|c|c|}
\hline
No.&$\mathcal{L}$& $C_0$& $C_1$& $C_{00}^{AB}$& $C_{01}^{AB}$&$C_{10}^{AB}$&$C_{11}^{AB}$& $C_{00}^{AC}$& $C_{01}^{AC}$&$C_{10}^{AC}$&$C_{11}^{AC}$\\
\hline
 1 & -3 & -2 & -2 & 2 & 2 & -1 & 1 & 0 & 1 & 0 & 0 \\
\hline
 2 & -4 & -2 & -4 & -2 & 2 & 2 & 2 & 1 & 0 & 0 & 1 \\
\hline
 3 & -5 & -3 & -3 & 2 & 2 & 2 & -3 & 1 & 0 & -1 & 2 \\
\hline
 4 & -6 & -4 & -6 & -3 & 2 & 3 & 2 & 2 & 0 & 1 & 1 \\
\hline
 5 & -11 & -4 & -12 & -4 & 6 & 6 & 6 & 1 & -1 & -1 & 4 \\
\hline
 6 & -7 & -5 & -5 & 2 & 3 & 2 & -4 & 1 & 1 & -1 & 3 \\
\hline
 7 & -8 & -6 & -8 & -4 & 3 & 3 & 2 & 3 & 1 & 1 & 1 \\
\hline
 8 & -5 & -2 & 2 & 2 & -2 & -2 & -4 & 1 & 1 & 1 & 2 \\
\hline
 9 & -3 & -3 & 1 & 1 & 1 & 1 & -1 & 1 & 0 & -1 & 1 \\
\hline
 10 & -6 & -4 & 2 & 2 & 2 & 2 & -4 & 1 & -1 & -1 & 3 \\
\hline
 11 & -6 & -6 & 0 & 2 & 3 & 3 & -2 & 3 & -1 & -1 & 1 \\
\hline
\end{tabular}
\end{center}
\caption{Inequalities for nearest- and next-to-nearest neighbor correlators}
\label{ineqs_list}
\end{table}

\begin{table}[htbp!]
\begin{center}
\begin{tabular}{|c|c|c|c|c|c|c|c|}
\hline
No.&$\mathcal{L}$&$\mathcal{Q}$&$\theta$&$\phi$&$\inf{\mathcal{Q}}$&$\inf{\mathcal{NS}}$&Genuine\\
\hline
 1 & -3 & -3.111 & 6.236 & 1.501 & -3.1907 & -3.5 & N \\
\hline
 2 & -4 & -4.184 & 0.077 & 1.874 & -4.38643 & -4.8 & N \\
\hline
 3 & -5 & -5.098 & 2.17 & 6.275 & -5.3502 & -5.8 & N \\
\hline
 4 & -6 & -6.179 & 6.236 & 4.175 & -6.35706 & -6.8 & Y \\
\hline
 5 & -11 & -11.104 & 5.996 & 4.691 & -11.7124 & -12.87 & N \\
\hline
 6 & -7 & -7.073 & 4.093 & 0.29 & -7.31685 & -7.8 & Y \\
\hline
 7 & -8 & -8.191 & 4.359 & 6.197 & -8.52433 & -9.06 & Y \\
\hline
 8 & -5 & -5.039 & 3.169 & 5.226 & -5.32177 & -5.8 & N \\
\hline
 9 & -3 & -3.04 & 3.843 & 1.193 & -3.22662 & -3.5 & Y \\
\hline
 10 & -6 & -6.109 & 0.817 & 2.421 & -6.37417 & -7 & Y \\
\hline
 11 & -6 & -6.081 & 3.787 & 6.067 & -6.36487 & -7 & Y \\
\hline
\end{tabular}
\end{center}
\caption{Violations of the inequalities under different scenarios.}
\label{ineqs_violations}
\end{table}

\section{Variational quantum optimizations of TI Bell inequalities}
\label{app:variational}

Given a Bell functional of the form

\be
B(P)\equiv\sum_{x,y,a,b}B_{x,y,a,b}P(a,b|x,y),
\ee

\noindent we wish to minimize its value over all distributions $P(a,b|x,y)$ of the form:

\be
P(a,b|x,y)=\tr(\rho_{AB} M_{x,a}\otimes M_{y,b}),
\label{primo}
\ee

\noindent where $\rho_{AB}$ is the 2-reduced density matrix of a TI state and $\{M_{x,a}\}_a$ are POVMs, i.e., $M_{x,a}\geq 0,\sum_{a}M_{x,a}=\id$.

Ideally, we would like to devise a sort of see-saw algorithm, that, starting from a random configuration of measurements $\{M_{x,a}\}_a$, would optimize over the state $\rho_{AB}$. Then, fixing the optimal $\rho_{AB}$, we optimize over the measurements $\{M_{x,a}\}_a$, and so on until the objective function converges.

For fixed measurements $\{M_{x,a}\}_a$, one just needs to optimize the Hamiltonian

\be
\sum_{x,y,a,b}B_{x,y,a,b} M_{x,a}\otimes M_{y,b}
\ee

\noindent over TI states $\rho_{AB}$; this can be done with Time Evolving Block-Decimation (TEBD)~\cite{PhysRevLett.98.070201}. For fixed $\rho_{AB}$, optimizing over the measurements is complicated, since the objective function is bilinear in them, i.e., it contains terms of the form $M_{x,a}\otimes M_{y,b}$.

Let us then try a different method: suppose that the TI state of the chain is of the form $\rho\otimes \omega$, where $\omega$, the register, is a classical state of the form $\omega=\sum_{s_1,s_2,\ldots=0}^{r-1}P(s_1,s_2,\ldots)\proj{s_1,s_2,\ldots}$, where $P(s_1,s_2,\ldots)$ is a TI probability distribution and $\rho$ is a TI state with site dimension $d$. It follows that the overall state $\rho\otimes \omega$ has site dimension $r\times d$.

The protocol that site $k$ will use to produce an outcome is as follows: first, he measures his register, obtaining a result $s\in\{0,\ldots,r-1\}$. Then, depending on his input $x_k$, he will conduct the measurement $\{M^{(s)}_{x,a}\}_a$ over the rest of his state, namely, $\rho_k$.

The statistics observed by two nearest neighbors are thus given by

\be
P(a,b|x,y)=\tr\{\rho_{AB}\left(\sum_{s_A,s_B=0}^{r-1}P(s_A,s_B)M^{(s_A)}_{x,a}\otimes M^{(s_B)}_{y,b}\right)\}.
\ee

Now, by (\ref{extreme}), the extreme points of bipartite distributions with a TI extension are of the form $(s,s)$ (deterministic) or a convex combination with equal weights of deterministic points $(s^{(i)},s^{(i+1)})$, $i=1,...,m$, with $1<m\leq r$, $s^{(m+1)}=s^{(1)}$ and such that $s^{(i)}\not=s^{(j)}$ for $i\not=j$, $i,j=1,...,m$. In the latter case, after the relabeling $E^{(s_i)}_{x,a}\to E^{(i)}_{x,a}$, we find that $P(a,b|x,y)$ can be rewritten as:

\be
P(a,b|x,y)=\tr\left\{\rho_{AB}\left(\frac{1}{m}\sum_{i=0}^{m-1}M^{(i)}_{x,a}\otimes M^{(i+1)}_{y,b}\right)\right\}.
\ee

\noindent This form also covers the first case, by choosing $M^{(i)}_{x,a}$ independent of $i$. Hence, in order to optimize $B(P)$, it suffices to consider expressions of the form

\be
B(P)=\sum_{x,y,a,b}B(x,y,a,b)\tr\left\{\rho_{AB}\left(\frac{1}{m}\sum_{i=0}^{m-1}M^{(i)}_{x,a}\otimes M^{(i+1)}_{y,b}\right)\right\},
\label{obj}
\ee

\noindent for $m=2,3,...,r$.

Note that now the objective function is linear on each of the measurement sets $\vec{M}^{(i)}\equiv\{M^{(i)}_{x,a}:x,a\}$, for $i=0,\ldots,m-1$. It follows that, for fixed $\rho_{AB}$ and fixed measurement sets $\{\vec{M}^{(j)}:j\not=i\}$, optimizing over $\vec{M}^{(i)}$ amounts to solving a semidefinite programming (SDP) problem~\cite{sdp}, and thus it can be carried out efficiently.

Hence, the full see-saw method proposed here to minimize $B(P)$ works at follows:

\begin{enumerate}
\item
Choose $m>1$.
\item
Generate random (extremal) measurement settings $\vec{E}^{(j)}$, perhaps by solving random SDP problems.
\item
Optimize $\rho_{AB}$ in (\ref{obj}) using TEBD.
\item
Fixing $\rho_{AB}$ and all sets $\vec{M}^{(j)}$ but $\vec{M}^{(i)}$, use SDP to optimize over (\ref{obj}) $\vec{M}^{(i)}$. repeat for $i=0,\ldots,m-1$.
\item
Go to step 3, iterating until the objective function seems to converge.

Technically, for a fixed size $r$ of the classical register, we should repeat this procedure for $m=2,...,r$. Since we are just interested in finding the ultimate quantum violations of the inequality, it suffices to choose $m$ large and possibly non-prime.

The scheme above was conceived to optimize over TI Bell inequalities involving only nearest-neighbor correlations. The extension to more parties is straightforward. For three parties, for instance, we just need to replace all instances of the expression $M_{x,a}\otimes M_{y,b}\otimes M_{z,c}$ in the Bell operator by

\be
\frac{1}{m}\sum_{i=0}^{m-1}M^{(i)}_{x,a}\otimes M^{(i+1)}_{y,b}\otimes M^{(i+2)}_{z,c},
\ee

\noindent where the superscripts should be taken modulo $m$, and $m\geq 3$.

\end{enumerate}

\end{appendix}

\end{document}